\documentclass[hyperref,envcountsame]{llncs}

\usepackage[svgnames]{xcolor}
\usepackage{amssymb}
\usepackage{amsmath}
\usepackage[mathcal]{euscript}
\usepackage[T1]{fontenc}
\usepackage[latin1]{inputenc}
\usepackage{ae,aecompl}
\usepackage{enumitem}
\usepackage{float}
\usepackage{xspace}
\usepackage{bm}
\usepackage{microtype}
\usepackage{pdfsync}
\usepackage[active]{srcltx}
\usepackage{relsize}
\usepackage{tikz}
\usetikzlibrary{trees,arrows,shapes,snakes,fit,shadows,decorations.text,decorations.pathmorphing,calc}

\let\doendproof\endproof
\newtheorem{clm}{Claim}
\renewcommand\endproof{~\hfill\qed\doendproof}

\let\leq\leqslant

\setenumerate[1]{label=$(\arabic*)$,leftmargin=*,topsep=0.3ex}
\setenumerate[2]{label=$(\alph*)$,leftmargin=*}
\setenumerate[3]{label=$(\roman*)$,leftmargin=*}
\setitemize[1]{label=$-$,leftmargin=*,topsep=0.3ex}
\setitemize[2]{label=$\circ$,leftmargin=*}
\setitemize[3]{label=$--$,leftmargin=*}

\newtheorem{myremark}{Remark}

\DeclareSymbolFont{bbold}{U}{bbold}{m}{n}
\DeclareSymbolFontAlphabet{\mathbbold}{bbold}
\newcommand\BM[1]{{\underline{\ensuremath{\boldsymbol #1}}}}

\newcommand\Sep{\ensuremath{\mathsf{Sep}}\xspace}

\newcommand\content[1]{\ensuremath{\contentmorphism(#1)}}
\newcommand\contentscc[2]{\ensuremath{\contentmorphism\_\mathsf{scc}(#1,#2)}}

\newcommand\contentmorphism{\ensuremath{\textsf{alph}}}
\newcommand\scc[2]{\ensuremath{\textsf{scc}(#1,#2)}}

\newcommand{\PT}{\ensuremath{\textsl{PT}}\xspace}
\newcommand{\kPT}[1]{\ensuremath{PT[#1]}\xspace}
\newcommand{\UL}{\ensuremath{\textsl{UL}}\xspace}

\newcommand{\ptclos}[2]{\ensuremath{[#1]_{\sim#2}}}
\newcommand{\fdclos}[2]{\ensuremath{[#1]_{\cong#2}}}

\newcommand{\bsi}{\ensuremath{\mathcal{B}\Sigma_{1}(<)}\xspace}
\newcommand{\siu}{\ensuremath{\Sigma_{1}(<)}\xspace}

\newcommand{\dew}{\ensuremath{\Delta_{2}(<)}\xspace}

\newcommand{\siw}{\ensuremath{\Sigma_{2}(<)}\xspace}

\newcommand\fodw{\ensuremath{\textup{FO}^2(<)}\xspace}

\newcommand\nat{\ensuremath{\mathbb{N}}\xspace}

\newcommand\pata[1]{\ensuremath{#1}-pattern\xspace}
\newcommand\patas[1]{\ensuremath{#1}-patterns\xspace}

\newcommand\kpata{\pata{k}}
\newcommand\kpatas{\patas{k}}

\newcommand\patt[2]{\ensuremath{(#1,#2)}-pattern\xspace}
\newcommand\patts[2]{\ensuremath{(#1,#2)}-patterns\xspace}

\newcommand\temp[1]{\ensuremath{#1}-template\xspace}
\newcommand\temps[1]{\ensuremath{#1}-templates\xspace}
\newcommand\ltemp{\temp{\ell}}
\newcommand\ltemps{\temps{\ell}}

\newcommand\impl[1]{\ensuremath{#1}-implementation\xspace}
\newcommand\impls[1]{\ensuremath{#1}-implementations\xspace}

\newcommand\kimpl{\impl{p}}
\newcommand\kimpls{\impls{p}}

\newcommand\impv[1]{\ensuremath{#1}-implements\xspace}
\newcommand\kimpv{\impv{p}}

\newcommand\peq[1]{\ensuremath{\sim_{#1}}\xspace}
\newcommand\kpeq{\peq{\kappa}}

\newcommand\fodeq[1]{\ensuremath{\cong_{#1}}\xspace}
\newcommand\kfodeq{\fodeq{\kappa}}

\newcommand\decop[1]{\ensuremath{#1}-decomposition\xspace}
\newcommand\decops[1]{\ensuremath{#1}-decompositions\xspace}
\newcommand\kdecop{\decop{k}}
\newcommand\kdecops{\decops{k}}

\def\pv#1{\ensuremath{{\sf#1}}}

\def\paragraph#1{\smallskip\noindent{\bfseries\itshape{#1}}}

\begin{document}
\title{Separating regular languages by piecewise
  testable and unambiguous languages}
\titlerunning{The separation problem for regular languages}

\author{Thomas~Place, Lorijn~van Rooijen \and Marc~Zeitoun}

\institute{LaBRI, Universit\'es de Bordeaux \& CNRS \\
  UMR~5800. 351 cours de la Lib\'eration, 33405 Talence Cedex, France.\\
  \email{tplace@labri.fr}, \email{lvanrooi@labri.fr}, \email{mz@labri.fr}
}

\maketitle

\begin{abstract}
  Separation is a classical problem asking whether, given two sets
  belonging to some class, it is possible to separate them by a set
  from a smaller class. We discuss the separation problem for regular
  languages. We give a \textsc{Ptime} algorithm to check whether two
  given regular languages are separable by a piecewise testable
  language, that is, whether a $\mathcal{B}\Sigma_1(<)$ sentence can
  witness that the languages are disjoint. The proof refines an
  algebraic argument from Almeida and the third author. When
  separation is possible, we also express a separator by saturating
  one of the original languages by a suitable congruence. Following
  the same line, we show that one can as well decide whether two
  regular languages can be separated by an
  unambiguous language, albeit with a higher complexity.
\end{abstract}

\section{Introduction}
\label{sec:introduction}

Separation is a classical
notion in mathematics and computer science. In general, one says
that two structures $L_1,L_2$ from a class $\mathcal{C}$ are \emph{separable} by a set $L$ if $L_1\subseteq
L$ and $L_2\cap L=\varnothing$. In this case, $L$ is called a
\emph{separator}. 

In separation problems, the separator $L$ is required to belong to a given subclass \Sep of $\mathcal{C}$. The problem asks whether two disjoint elements $L_1, L_2$ of
$\mathcal{C}$ can always be separated by an element of the subclass
\Sep.

In the case that disjoint elements of $\mathcal{C}$ cannot always be separated
by an element of \Sep, several natural questions arise:
\begin{enumerate}
\item\label{item:1} given elements $L_1, L_2$ in $\mathcal{C}$, can we decide whether a separator exists in~\Sep?
\item if so, what is the complexity of this decision problem?
\item\label{item:3} can we, in addition, compute a separator, and what is the complexity?
\end{enumerate}
In this context, it is known for example that separation of two
context-free languages by a regular one is
undecidable~\cite{Hunt:1982:DGP:322307.322317}.  

\paragraph{Separating regular languages.}
In this paper, we look at separation problems for the class
$\mathcal{C}$ of regular languages. 
These separation problems generalize the task of finding decidable characterizations for 
subclasses \Sep of $\mathcal{C}$ which are closed under complement: a separation algorithm for a subclass \Sep 
entails an algorithm for deciding membership in \Sep (\emph{i.e.}, membership
reduces to separability). Indeed, membership in
\Sep can be checked by testing whether the input language
is \Sep-separable from its complement. 

While finding a decidable characterization for \Sep already requires a
deep understanding of the subclass, the search for separation
algorithms is intrinsically more difficult.  Indeed,
 powerful tools are already available to decide membership in \Sep: one normally makes use of a recognizing device of the input language, \emph{viz}.~its syntactic monoid. A famous result along these lines is Sch\"utzenberger's Theorem~\cite{saperiodic}, which
states that a language is definable in first-order logic if and only if its
syntactic monoid is aperiodic, which is easily decidable.

Now for a separation algorithm, the question is whether the input languages are \emph{sufficiently different}, from the point of view of the subclass \Sep, to allow this to be witnessed by an element of \Sep. 
Note that we cannot use standard methods on the recognizing devices,
as was the case for the membership problem. We now have to decide
whether there \emph{exists} a recognition device of the given type
that separates the input: we do not have it in hand, nor its syntactic monoid.
An even harder question then is to actually construct the so-called separator in \Sep.  

\paragraph{\textbf{Contributions.}}
In this paper, we study this problem for two subclasses of the regular
languages: piecewise testable languages and unambiguous languages.

Piecewise testable languages are languages that can be described by the presence or absence of scattered subwords up to a certain size
within the words. Equivalently, these are the languages 
definable using \bsi formulas, that is, first-order logic formulas that are boolean
combinations of \siu formulas. A \siu formula is 
a first-order formula with a quantifier prefix $\exists^{*}$. A well-known result about piecewise
testable languages is Simon's Theorem~\cite{Simon:1975:PTE:646589.697341} that states that a
regular language is piecewise testable if and only if its syntactic monoid is
$\mathcal{J}$-trivial. This property yields a decision procedure to check  whether a language is piecewise testable. Stern has refined this procedure into a polynomial time algorithm~\cite{Stern:Complexity-some-problems-from:1985:a}, of which the complexity has been improved by Trahtman~\cite{Trahtman:Piecewise-Local-Threshold-Testability:2001:b}.

The second class that we consider is the class of unambiguous
languages, \emph{i.e.}, languages defined by unambiguous products. This class has been given many equivalent
characterizations~\cite{Tesson02diamondsare}. For example, these
are the \fodw-definable languages, \emph{i.e.},~languages that can be
defined in first-order logic using only two variables. Equivalently,
this is the class \dew of languages that are definable by a
first-order formula with a quantifier prefix $\exists^{*}\forall^{*}$
and simultaneously by a first-order formula with a quantifier prefix
$\forall^{*}\exists^{*}$. Note that consequently, all piecewise
testable languages are \fodw-definable. It has been shown
in~\cite{pwdelta} for \dew, and in~\cite{twfodeux} for \fodw that these
are exactly the languages whose syntactic monoid belongs to the decidable class~$\pv{DA}$.

\smallskip
There is a common difficulty in the separation problems for these two
classes. \emph{A priori}, it is not known up to which level one should
proceed in refining the candidate separators to be able to answer the
question of separability. For piecewise testable languages, this
refinement basically means increasing the size of the considered
subwords. For unambiguous languages this means increasing the
size of the unambiguous products. 
For both of these classes, we are able to compute, from the two input languages, a number that suffices for this purpose. 
This entails decidability of the separability problem for both
classes.

A rough analysis yields a 3-\textsc{Nexptime} upper bound for
separation by unambiguous languages. For separability by piecewise
testable languages, we obtain a better bound starting from NFAs: we show that two
languages are separable if and only if the corresponding automata do not contain
certain forbidden patterns of the same type, and we prove that
the presence of such patterns can be decided in polynomial time
wrt.~the size of the automata and of the alphabet. This yields a
\textsc{Ptime} algorithm for deciding separation by a piecewise
testable language.

\paragraph{\textbf{Related work.}} 
The classes of piecewise testable and unambiguous languages are varieties of regular languages. For such varieties, there is a generic
connection found by Almeida~\cite{MR1709911} between profinite
semigroup theory and the separation problem: Almeida has shown that two regular languages over $A$
are separable by a language of a variety $A^*\mathcal{V}$ if and only if the
topological closures of these two languages inside a profinite
semigroup, depending only on $A^*\mathcal{V}$, intersect. 
Note that this theory does not give any information about how to actually \emph{construct} the separator, in case two languages are separable. 
To turn Almeida's result into an algorithm deciding separability, we 
should  compute representations of these topological closures, and test for emptiness of intersections of such~closures. 

So far, these problems have no generic answer and have been studied in an algebraic context
for a small number of specific varieties. Deciding whether the closures of two regular languages
intersect is equivalent to computing the so-called 2-pointlike sets of
a finite semigroup wrt.~the considered variety,
see~\cite{MR1709911}. This question has been answered positively 
for the varieties of finite group languages~\cite{ash:1991:a,ribes&zalesskii:1993:a}, piecewise testable languages~\cite{MR1611659,MR2365328}, star-free languages~\cite{DBLP:journals/ijac/HenckellRS10a,Henckell:1988}, 
and a few other varieties, but it was left open for unambiguous languages. 

A general issue is that the topological closures may not be describable
by a finite device. However, for 
piecewise testable languages, the approach of~\cite{MR1611659}
builds an automaton over an extended alphabet, which
recognizes the closure of the original language. This can be performed
in polynomial time wrt.\ the size of the original automaton. Since
these automata admit the usual construction for intersection and can
be checked for emptiness in \textsc{Nlogspace}, this gives a polynomial time
algorithm wrt.~the size of the original automata. The construction was
presented for deterministic automata but also works for
nondeterministic ones. One should mention that the extended alphabet
is $2^A$ (where $A$ is the original alphabet). 
Therefore, these results give an algorithm which, from two NFAs,
decides separability by piecewise testable languages in time
polynomial in the number of states of the NFAs and exponential in the
size of the original~alphabet. 

Our proof for separability by piecewise testable languages follows the same pattern as the method described above, but a significant improvement is that we show that non-separability is witnessed by both
automata admitting a path of the same shape. This allows us to present an algorithm that provides better complexity as it runs in
polynomial time in both the size of the automata, \emph{and} in the
size of the alphabet. Also, we do not make use
of the theory of profinite semigroups: we work only with
elementary concepts. We have described this algorithm in~\cite{vRZ}. 
Furthermore, we show how to compute from the input languages, an index that suffices to separate them. 
Recently, Martens et.~al.~\cite{martens} also provided a \textsc{Ptime} algorithm for deciding separability by piecewise testable
languages, using different proofs but do not provide the
computation of such an index.

Finally, for separation by unambiguous languages, the positive decidability
result of this paper is new, up to the authors' knowledge. It is
equivalent to the decidability of computing the 2-pointlike sets for the
class \pv{DA}.

\section{Preliminaries}
\label{sec:prelim}

We fix a finite alphabet  $A=\{a_1,\dots,a_m\}$. We denote by $A^*$
the free monoid over $A$. The empty word is denoted by~$\varepsilon$. For a word $u \in A^*$, the smallest
$B\subseteq A$ such that $u\in B^*$ is called the \emph{alphabet} of
$u$ and is denoted by $\content u$.

\paragraph{\textbf{Separability}.}
Given languages $L,L_1,L_2$, we say that $L$ \emph{separates} $L_1$
from $L_2$ if 
\begin{equation*}
L_1 \subseteq L \text{ and } L_2 \cap L = \varnothing.
\end{equation*}
Given a class \Sep of languages, we say that the pair $(L_1,L_2)$ is
\emph{\Sep-separable} if some language $L\in\Sep$ separates $L_1$ from
$L_2$. Since all classes we consider are closed under complement,
$(L_1,L_2)$ is \Sep-separable if and only if $(L_2,L_1)$ is, in which
case we simply say that $L_1$ and $L_2$ are \Sep-separable.

We are interested in two classes \Sep  of separators: the class of
piecewise testable languages, and the class of unambiguous
languages.

\paragraph{\textbf{Piecewise Testable Languages}.}
We say that a word $u$ is a \emph{piece} of $v$, if
\vspace*{-1mm}
\begin{equation*}
  u=b_1\cdots b_k,\text{ where } b_1,\ldots ,b_k\in A,\text{\quad and }
  v\in A^* b_1A^*\cdots A^*b_kA^*.
\end{equation*}
For instance, $ab$ is a piece of $bb\BM{a}cc\BM{b}a$. The \emph{size}
of a piece is its number of letters.
A language $L\subseteq A^*$ is \emph{piecewise testable} if there
exists a $\kappa \in \nat$ such that membership of $w$ in $L$ only
depends on the pieces of size up to $\kappa$ occurring in~$w$.
We write $w \kpeq w'$ when $w$ and $w'$ have the same
pieces of size up to $\kappa$. Clearly, \kpeq is an equivalence
relation, and it has finite index (since there are finitely many
pieces of size $\kappa$ or less). Therefore, a language is piecewise
testable if and only if it is a union of \kpeq-classes for
some~$\kappa\in\nat$.  In this case, the language is said to be of
\emph{index}~$\kappa$.  It is easy to see that a language is piecewise
testable if and only if it is a finite boolean combination of
languages of the form $A^*b_1A^*\cdots A^*b_kA^*$.

Piecewise testable languages are those definable by \bsi
formulas. \bsi formulas are boolean combinations of first-order
formulas of the form:
\[
\exists x_{1}\ldots \exists x_{n}\ \varphi(x_{1},\ldots,x_{n}),
\]
where $\varphi$ is quantifier-free. For instance, $A^*b_1A^*\cdots A^*b_kA^*$ is defined by the formula
$\exists x_1\ldots\exists x_k \bigl[\bigwedge_{i<k}
(x_i<x_{i+1})\land\bigwedge_{i\leq k} b_i(x_i)\bigr]$, where the first-order variables $x_1,\ldots, x_k$ are interpreted as positions, and
where $b(x)$ is the predicate testing that position $x$ carries
letter~$b$.

We denote by $\PT[\kappa]$ the class of all piecewise testable
languages of index $\kappa$ or less, and by $\PT=\bigcup_\kappa\PT[\kappa]$
the class of all piecewise testable languages.

Given $L\subseteq A^*$ and $\kappa\in\nat$, the smallest
$\PT[\kappa]$-language containing~$L$ is
$$\ptclos L\kappa=\{w\in A^*\mid \exists u\in L\text{ and }u\kpeq w\}.$$
In general however, there is no smallest \PT-language containing a
given language, because removing one word from a \PT-language yields
again a \PT-language.

\paragraph{\textbf{Unambiguous Languages.}} A product
$L=B_0^*a_1B_1^*\cdots B_{k-1}^*a_kB_k^*$ is called \emph{unambiguous}
if every word of $L$ admits exactly one factorization witnessing its
membership in $L$. The integer $k$ is called the \emph{size} of the
product. An \emph{unambiguous language} is a finite disjoint
union of unambiguous products. Observe that unambiguous languages are
connected to piecewise testable languages. Indeed, it was
proved in~\cite{schul} that the class of unambiguous languages is
closed under boolean operations. Moreover, languages of the
form $A^*b_1A^*\cdots A^*b_kA^*$ are unambiguous, witnessed by the
product $(A\setminus\{b_1\})^*b_1(A\setminus\{b_2\})^*\cdots
(A\setminus\{b_k\})^*b_kA^*$. Therefore, piecewise testable languages
form a subclass of the class of unambiguous~ones.

Many equivalent
characterizations for unambiguous languages have been found~\cite{Tesson02diamondsare}. From a logical point of
view, unambiguous languages are exactly the languages definable by an
\fodw formula~\cite{twfodeux}. \fodw is the two-variable restriction
of first-order logic. Another logical characterization which further
illustrates the link with piecewise testable languages
(i.e. \bsi-definable languages) is \dew. A \siw formula is a
first-order formula of the form:
\[
\exists x_{1}\ldots\exists x_{n}\ \forall y_{1}\ldots\forall y_{m}\ \varphi(x_{1},\ldots,x_{n},y_{1},\ldots,y_{m}),
\]
where $\varphi$ is quantifier-free. A language is
\emph{\dew-definable} if it can be defined both by a \siw formula and
the negation of a \siw formula. It has been proven in~\cite{pwdelta}
that a language is unambiguous if and only if it is \dew-definable.

For two words $w,w'$, we write, $w \kfodeq w'$ if $w,w'$ satisfy the
same unambiguous products of size $\kappa$ or less. We denote by $\UL[\kappa]$ the
class of all languages that are unions of equivalence classes of
$\kfodeq$, and we let $\UL=\bigcup_\kappa\UL[\kappa]$. Observe that because
unambiguous languages are closed under boolean operations, \UL is the
class of all unambiguous languages.

Given $L\subseteq A^*$ and $\kappa\in\nat$, the smallest
$\UL[\kappa]$-language containing~$L$ is
$$\fdclos L\kappa=\{w\in A^*\mid \exists u\in L\text{ and }u\kfodeq w\}.$$
In general again, there is no smallest \UL-language
containing a given language.

\paragraph{\textbf{Automata.}}
\label{paragraph:automata}
A \emph{nondeterministic finite automaton} (NFA) over $A$ is denoted
by a tuple $\mathcal{A}=(Q,A,I,F,\delta)$, where $Q$ is the set of
states, $I\subseteq Q$ the set of initial states, $F\subseteq Q$ the
set of final states and $\delta\subseteq Q\times A\times Q$ the
transition relation. The \emph{size} of an automaton is its number of
states plus its number of transitions.  We denote by $L(\mathcal{A})$
the language of words accepted by $\mathcal{A}$.  Given a word $u\in
A^*$, a subset $B$ of $A$ and two states $p,q$ of $\mathcal{A}$, we
denote
\begin{itemize}
\item by $p\xrightarrow{\ u\ }q$ a path from state $p$ to state $q$ labeled $u$. 
\item by $p\xrightarrow{{}\subseteq B} q$ a path from $p$ to $q$ of
  which all transitions are labeled over~$B$.
\item by $p\xrightarrow{{}=B}q$ a path from $p$ to $q$ of which all
  transitions are labeled over $B$, with the additional demand that
  every letter of $B$ occurs at least once along it.
\end{itemize}
Given a state $p$, we denote by $\scc{p}{\mathcal{A}}$ the strongly
connected component of $p$ in $\mathcal{A}$ (that is, the set of
states that are reachable from $p$ and from which $p$ can be reached),
and by $\contentscc{p}{\mathcal{A}}$ the set of labels of all
transitions occurring in this strongly connected component. Finally,
we define the restriction of $\mathcal{A}$ to a subalphabet
$B\subseteq A$ by
$\mathcal{A}\restriction_B\mathrel{\;\stackrel{\text{def}}=}(Q,B,I,F,\delta\cap
(Q\times B\times Q))$. 

\paragraph{\textbf{Monoids.}}
Let $L$ be a language and $M$ be a monoid. We say that \emph{$L$ is
  recognized by $M$} if there exists a monoid morphism $\alpha : A^*
\rightarrow M$ together with a subset $F \subseteq M$ such that $L =
\alpha^{-1}(F)$. It is well known that a language is accepted by an
NFA if and only if it can be recognized by a
\emph{finite~monoid}. Further, one can compute from any NFA a finite
monoid recognizing its accepted~language.

\section{Separation by piecewise testable languages}
\label{sec:separation}

Since $\PT[\kappa]\subset\PT$, $\PT[\kappa]$-separability implies
$\PT$-separability.  Further, for a fixed~$\kappa$, it is obviously
decidable whether two languages $L_1$ and $L_2$ are
$\PT[\kappa]$-separable: there is a finite number of $\kPT\kappa$
languages over $A$, and for each of them, one can test whether it
separates $L_1$ and $L_2$.  The difficulty for deciding whether $L_1$
and $L_2$ are $\PT$-separable is to effectively compute a witness
$\kappa=\kappa(L_1,L_2)$, \emph{i.e.}, such that $L_1$ and $L_2$ are
\PT-separable if and only if they are $\kPT\kappa$-separable.
Actually, we show that $\PT$-separability is decidable, by 
different arguments:
\begin{enumerate}[labelsep=1.5ex,leftmargin=6ex]
\item[$(1.a)$] We give a necessary and sufficient condition  on
  NFAs recognizing $L_1$ and $L_2$, in terms of
  forbidden patterns, to test whether $L_1$ and $L_2$ are
  \PT-separable.
\item[$(1.b)$] We give a polynomial time algorithm to check this condition.
\item[$(2)$] We compute $\kappa\in\nat$ from $L_1,L_2$, such that
  $\PT$-separability and $\kPT\kappa$-separability 
  are equivalent for $L_1$ and $L_2$.  Hence, if the \textsc{Ptime} algorithm answers that
  $L_1$ and $L_2$ are \PT-separable, then $\ptclos{L_1}{\kappa}$ is a
  valid \hbox{\PT-separator}.
\end{enumerate}

Let us first introduce some terminology to
explain the necessary and sufficient condition on NFAs. Let
$\mathcal{A}$ be an NFA over~$A$. For $u_0,\ldots,u_p\in A^*$ and
\emph{nonempty} subalphabets $B_1,\ldots,B_p\subseteq A$, let
$\vec u=(u_0,\ldots,u_p)$ and $\vec B=(B_1,\ldots,B_p)$. We call
$(\vec u, \vec B)$ a \emph{factorization pair}.  A \emph{$(\vec u,\vec
  B)$-path} in $\mathcal{A}$ is a successful path (leading from the
initial state to a final state of $\mathcal{A}$), of the form shown in
\figurename~\ref{fig:ub-path}.

\tikzstyle{nod}=[minimum size=0.3cm,draw,circle,inner sep=2pt]
\tikzstyle{nof}=[minimum size=0.3cm,draw,circle,double]
\tikzstyle{ars}=[line width=0.7pt,->]

\tikzstyle{arr}=[line width=2.0pt,double,->]
\tikzstyle{wrd}=[line width=0.25cm,gray!50,-]
\tikzstyle{sep}=[line width=1.0pt,-]
\tikzstyle{zig}=[line width=1.0pt,snake=zigzag]
\tikzstyle{ond}=[line width=1.0pt,snake=snake]
\tikzstyle{acc}=[line width=1.0pt,snake=brace]
\tikzstyle{stp}=[midway,draw,thick,rounded rectangle,align=center,fill=white]

\tikzstyle{bag}=[inner sep=0pt]
\tikzstyle{box}=[rectangle]
\vspace*{-3ex}
\begin{figure}[H]
  \centering
    \begin{tikzpicture}
      \node[nod] (q0) at (0.0,0.0) {};
      \node[nod] (q1) at ($(q0)+(1.0,0.0)$) {};
      \node[nod] (q2) at ($(q1)+(1.6,0.0)$) {};
      \node[nod] (q3) at ($(q2)+(1.6,0.0)$) {};
      \node[nod] (q4) at ($(q3)+(1.0,0.0)$) {};
      \node[nod] (q5) at ($(q4)+(0.8,0.0)$) {};
      \node[nod] (q6) at ($(q5)+(1.0,0.0)$) {};
      \node[nod] (q7) at ($(q6)+(1.6,0.0)$) {};
      \node[nod] (q8) at ($(q7)+(1.6,0.0)$) {};
      \node[nof] (q9) at ($(q8)+(1.0,0.0)$) {};

      \draw[ars] ($(q0)-(0.7,0.0)$) to (q0);

      \draw[ars] (q0) to node[above] {$u_0$} (q1); 
      \draw[ars] (q1) to node[above] {${}\subseteq B_1$} (q2); 
      \draw[ars] (q2) to node[above] {${}\subseteq B_1$} (q3); 
      \draw[ars] (q3) to node[above] {$u_1$} (q4); 
      \draw[very thick,dotted] (q4) to (q5); 
      \draw[ars] (q5) to node[above] {$u_{p-1}$} (q6); 
      \draw[ars] (q6) to node[above] {${}\subseteq B_p$} (q7); 
      \draw[ars] (q7) to node[above] {${}\subseteq B_p$} (q8); 
      \draw[ars] (q8) to node[above] {$u_p$} (q9); 

      \draw[ars] (q2) to [out=115,in=65,loop] node[above] {${}= B_1$} (q2);
      \draw[ars] (q7) to [out=115,in=65,loop] node[above] {${}= B_p$} (q7);
    \end{tikzpicture}
  \caption{A $(\vec u, \vec B)$-path}
  \label{fig:ub-path}
\end{figure}
\vspace*{-3ex}
Recall that edges denote sequences of transitions (see 
section \emph{Automata}, p.~\pageref{paragraph:automata}).
Therefore, if $\mathcal{A}$
has a $(\vec u,\vec B)$-path, then $L(\mathcal{A})$
contains a language of the form $u_0(x_1y_1^*z_1)u_1\cdots
u_{p-1}(x_py_p^*z_p)u_p$, where
$\content{x_i}\cup\content{z_i}\subseteq \content{y_i}=B_i$.

Given NFAs $\mathcal{A}_1$ and $\mathcal{A}_2$, we say that a
factorization pair $(\vec u,\vec B)$ is a \emph{witness of non
  \PT-separability} for $(\mathcal{A}_1,\mathcal{A}_2)$ if there is a
$(\vec u,\vec B)$-path in both $\mathcal{A}_1$ and $\mathcal{A}_2$.
For instance, $\mathcal{A}_1$ and $\mathcal{A}_2$ pictured
in~\figurename~\ref{fig:ub-witness} have a witness of non
\PT-separability, namely the factorization pair $(\vec u,\vec B)$ with
$\vec u=(\varepsilon,c,\varepsilon)$ and $\vec B=(\{a,b\},\{a\})$.
%\newpage
\begin{figure}[h]
  \centering
    \begin{tikzpicture}
      %% A1
      \node[nod] (q0) at (0.0,0.0) {};
      \node[nod] (q1) at ($(q0)+(1.5,0.0)$) {};
      \node[nod] (q2) at ($(q1)+(1.5,0.0)$) {};
      \node[nof] (q3) at ($(q2)+(1.5,0.0)$) {};

      \draw[ars] ($(q0)-(0.5,0.0)$) to (q0);
      
      \draw ($(q1)+(0.75,-.7)$) node  {Automaton $\mathcal{A}_1$};
      \draw[ars, bend left] (q0) to node[above] {$a$} (q1); 
      \draw[ars, bend left] (q1) to node[below] {$b$} (q0); 
      \draw[ars] (q1) to node[below] {$c$} (q2); 
      \draw[ars] (q2) to node[below] {$a$} (q3);

      \draw[ars] (q2) to [out=115,in=65,loop] node[above] {$a$} (q2);

      %% A2
      \node[nod] (p0) at ($(q3)+(2,0.0)$) {};
      \node[nod] (p1) at ($(p0)+(1.5,0.0)$) {};      
      \node[nod] (p2) at ($(p1)+(1.5,0.0)$) {};      
      \node[nof] (p3) at ($(p2)+(1.5,0.0)$) {};      

      \draw[ars, bend left] (p0) to node[above] {$b$} (p1); 
      \draw[ars, bend left] (p1) to node[below] {$a$} (p0); 
      \draw[ars] (p1) to node[below] {$b$} (p2); 
      \draw[ars] (p2) to node[below] {$c$} (p3); 
      
      \draw[ars] ($(p0)-(0.5,0.0)$) to (p0);

      \draw ($(p1)+(0.75,-.7)$) node  {Automaton $\mathcal{A}_2$};
      \draw[ars] (p3) to [out=115,in=65,loop] node[above] {$a$} (p3);
    \end{tikzpicture}
  \caption{A witness of non \PT-separability for $(\mathcal{A}_1,\mathcal{A}_2)$: $\vec u=(\varepsilon,c,\varepsilon)$, $\vec B=(\{a,b\},\{a\})$}
  \label{fig:ub-witness}
\end{figure}
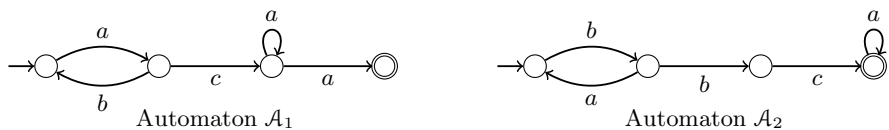
\vspace*{-3ex}
\noindent
We are now ready to state our main result regarding \PT-separability.
\begin{theorem}
  \label{thm:seppiece}
  Let $\mathcal{A}_1$ and $\mathcal{A}_2$ be two NFAs over 
  $A$. Let $L_1=L(\mathcal{A}_1)$ and $L_2=L(\mathcal{A}_2)$. Let $k_1,k_2$ be the number of states
  of $\mathcal{A}_1$ resp.~$\mathcal{A}_2$. Define
  $p=\max(k_1,k_2)+1$ and $\kappa=p|A|2^{2^{|A|}|A|(p|A|+1)}$. Then the following conditions are equivalent:
  \begin{enumerate}[leftmargin=7em]
  \item\label{item:6} $L_1$ and $L_2$ are \PT-separable.
  \item\label{item:8} $L_1$ and $L_2$ are $\PT[\kappa]$-separable.
  \item\label{item:7} The language $\ptclos {L_1}\kappa$ separates $L_1$ from $L_2$.
  \item\label{item:9} There is no witness of non \PT-separability in $(\mathcal{A}_1,\mathcal{A}_2)$.
  \end{enumerate}
\end{theorem}
Condition~\ref{item:8} yields an algorithm to test \PT-separability of
regular languages. Indeed, one can effectively compute all piecewise
testable languages of index $\kappa$ (of which there are finitely many), and for each of them, one can test whether it separates $L_1$
and $L_2$.  Before proving Theorem~\ref{thm:seppiece}, we show that
Condition~\ref{item:9} can be tested in polynomial time (and hence,
\PT-separability is \textsc{Ptime} decidable).

\begin{proposition}
  \label{prop:common-ub-path}
  Given two NFAs $\mathcal{A}_1$ and $\mathcal{A}_2$, one can determine
  whether there   exists a witness of non \PT-separability in
  $(\mathcal{A}_1,\mathcal{A}_2)$ in
  polynomial time wrt.~the sizes of $\mathcal{A}_1$ and $\mathcal{A}_2$, and the size
  of the alphabet.
\end{proposition}

\begin{proof}
  Let us first show that the following problem is in \textsc{Ptime}:
  given   states $p_1,q_1,r_1$ of $\mathcal{A}_1$ and $p_2,q_2,r_2$ of
  $\mathcal{A}_2$, determine whether there exists a nonempty
  $B\subseteq A$ and paths
  $p_i\xrightarrow{\subseteq B}q_i\xrightarrow{(= B)}q_i\xrightarrow{\subseteq B}r_i$ in
  $\mathcal{A}_i$ for both $i=1,2$.

  To do so, we compute a decreasing sequence $(C_i)_i$ of alphabets
  overapproximating the greatest alphabet $B$ that can be chosen for
  labeling the loops. Note that if there exists such an alphabet $B$, it should be contained in
  $$C_1\stackrel{\text{def}}{=}\contentscc{q_1}{\mathcal{A}_1} \cap \contentscc{q_2}{\mathcal{A}_2}.$$
  Using Tarjan's algorithm to compute  strongly connected components
  in linear time, 
  one can compute~$C_1$ in linear time as well. Then, we restrict the automata to alphabet $C_1$, and we repeat the process to obtain the sequence~$(C_i)_i$:
  \[
  C_{i+1} \stackrel{\text{def}}{=} \contentscc{q_1}{\mathcal{A}_1\restriction_{C_i}} \cap \contentscc{q_2}{\mathcal{A}_2\restriction_{C_i}}.
  \]
  After a finite number $n$ of iterations, we obtain $C_n =  C_{n+1}$. Note that $n\leq|\content{\mathcal{A}_1} \cap  \content{\mathcal{A}_2}|\leq |A|$. If $C_n = \varnothing$, then there exists  no nonempty $B$ for which there is an ($=B$)-loop around both $p_1$ and  $p_2$. If $C_n \ne \varnothing$, then it is the maximal nonempty  alphabet $B$ such that there are $(=B)$-loops around $q_1$ in  $\mathcal{A}_1$ and $q_2$ in $\mathcal{A}_2$. It then remains to determine  whether there exist paths $p_1\xrightarrow{{}\subseteq   B}q_1\xrightarrow{{}\subseteq B}r_1$ and $p_2\xrightarrow{{}\subseteq   B}q_2\xrightarrow{{}\subseteq B}r_2$, which can be performed in linear  time.

  To sum up, since the number $n$ of iterations to get
  $C_n=C_{n+1}$ is bounded by $|A|$, and since each computation is
  linear wrt.~the size of $\mathcal{A}_1$ and $\mathcal{A}_2$, one can
  decide in \textsc{Ptime} wrt.~to both $|A|$ and these sizes
  whether such pair of paths~occurs.

  Now we build from $\mathcal{A}_1$ and $\mathcal{A}_2$ two new
  automata   $\tilde{\mathcal{A}_1}$ and $\tilde{\mathcal{A}_2}$  as
  follows. The procedure first initializes $\tilde{\mathcal{A}_i}$ as
  a   copy of $\mathcal{A}_i$. Denote by $Q_i$ the state set of
  $\mathcal{A}_i$. For each   4-uple $\tau=(p_1,r_1,p_2,r_2)\in
  Q_1^2\times Q_2^2$ such that there   exist~$B\neq\varnothing$, two
  states $q_1\in Q_1,q_2\in Q_2$ and paths
  $p_i\xrightarrow{{}\subseteq
    B}q_i\xrightarrow{{}=B}q_i\xrightarrow{{}\subseteq B}r_i$ both for
  $i=1$ and $i=2$, we add in both $\tilde{\mathcal{A}_1}$ and
  $\tilde{\mathcal{A}_2}$     a new letter $a_\tau$ to the alphabet,
  and   ``summary'' transitions $p_1\xrightarrow{a_\tau}r_1$ and
  $p_2\xrightarrow{a_\tau}r_2$. Since there is a   polynomial number
  of tuples $(p_1,q_1,r_1,p_2,q_2,r_2)$, the above shows that
  computing these new transitions can be performed in \textsc{Ptime}. So, computing $\tilde{\mathcal{A}_1}$ and   $\tilde{\mathcal{A}_2}$ can be done in \textsc{Ptime}.

 By construction, there exists some factorization pair $(\vec u,
 \vec B)$ such that $\mathcal{A}_1$ and $\mathcal{A}_2$ both have a
 $(\vec u, \vec B)$-path if and only if $L(\tilde{\mathcal{A}_1})\cap
 L(\tilde{\mathcal{A}_2})\not=\varnothing$. Since   both
 $\tilde{\mathcal{A}_1}$ and $\tilde{\mathcal{A}_2}$ can be built in
 \textsc{Ptime}, this can be decided in polynomial time as well.
\end{proof}

The following is an immediate consequence of
Theorem~\ref{thm:seppiece} and Proposition~\ref{prop:common-ub-path}.
\begin{corollary}
  \label{thm:PT-separation}
  Given two NFAs, one can determine in polynomial time, with respect to the number of states and the size of the alphabet, whether the languages recognized by these NFAs are \PT-separable.\qed
\end{corollary}
In the rest of the section, we sketch the proof  of
Theorem~\ref{thm:seppiece}. The implications
\ref{item:7}$\Longleftrightarrow$\ref{item:8}$\implies$\ref{item:6}
are obvious. To show \ref{item:6}$\implies$\ref{item:8}, we introduce some
terminology.
Let us fix an arbitrary order $a_1<\cdots<a_m$ on~$A$.

\smallskip\noindent
\textbf{\patts pB.}  Let
$B=\{b_1,\ldots,b_r\} \subseteq A$ with $b_1<\cdots<b_r$, and let $p
\in \nat$. We say that a word $w \in A^*$ is a \emph{\patt{p}{B}} if
$w\in (B^*b_1B^*\cdots B^*b_rB^*)^p$. 
The number $p$ is called the
\emph{power} of~$w$.
For example, set $B=\{a,b,c\}$ with $a<b<c$. The word
$bb\BM{a}a\underline{\boldsymbol
  b}ab\BM{c}c\underline{\boldsymbol
  a}c\BM{b}aba\BM{c}a$ is a
\patt{2}{B} but not a \patt{3}{B}.

\smallskip\noindent {\bf \ltemps.} An \emph{\ltemp} is a sequence of length
$\ell$, $T=t_1,\dots, t_\ell$, such that every $t_i$ is either a letter
$a$ or a subset $B$ of the alphabet $A$. The main idea behind \ltemps
is that they yield decompositions of words that can be detected using
pieces and provide a suitable decomposition for
pumping. Unfortunately, not all \ltemps are actually detectable.
Because of this we restrict ourselves to a special case of \ltemps. An
\ltemp is said to be \emph{unambiguous} if all pairs $t_i,t_{i+1}$
are either two letters, two incomparable sets or a set and a letter
that is not included in the set. For example, $T=a,\{b,c\},d,\{a\}$ is
unambiguous, while $T'={\BM b},\{{\BM b},c\},d,\{a\}$ and $T'' =
a,\{b,{\BM c}\},{\{\BM c\}},\{a\}$ are not.

\smallskip\noindent
{\bf \kimpls.}  
A word $w\in A^*$ is a \emph{\kimpl} of an \ltemp $T=t_1,\dots, t_\ell$ if
$w=w_1 \cdots w_\ell$ and for all $i$ either $t_i=w_i\in A$ or $t_i=B \subseteq A$, $w_i \in B^{*}$ and $w_i$
is a \patt{p}{B}.
For example, $abccbbcbdaaaa=a.(\BM b\BM cc\BM bb\BM cb).d.(\BM a\BM aaa)$ is  a \impl{2} of the
\temp{4} $T=a,\{b,c\},d,\{a\}$, since $\BM b\BM cc\BM bb\BM cb$ is
a \patt{2}{\{b,c\}} and $\BM a\BM aaa$ is a \patt{2}{\{a\}}.

We now use \kimpls to prove \ref{item:6}$\implies$\ref{item:8}. The
proof is divided in two steps. First, we prove that there exists $p$
such that if two words are \kimpls of the same \ltemp for some $\ell$,
then they can be pumped into words containing the same pieces of size
$k$ for any $k$, while keeping membership in the regular languages. We will then prove that if two words contain the
same pieces for a large enough size, they are both \kimpls of a
common unambiguous \ltemp. We begin with the first step in the
following lemma.

\begin{lemma} \label{lem:piecepump}
Let $w_1 \in L_1$ and $w_2 \in L_2$. From $L_1, L_2$, we can compute $p \in \mathbb{N}$
such that whenever $w_1, w_2$ 
are both \kimpls of an \temp{\ell}
$T$ for some $\ell$, then for every $\kappa \in \mathbb{N}$, there exist $w'_1 \in
L_1$ and $w'_2 \in L_2$ such that $w'_1 \peq{\kappa} w'_2$.
\end{lemma}

\begin{proof}
  This is a pumping argument. Let $k_1,k_2$ be the number of states
  of automata recognizing $L_1$ and $L_2$ and set $p =\max(k_1,k_2)$.
  Set $w_1,w_2$ and $T=t_1,\dots,t_\ell$ as in the statement of the lemma.
  Fix $\kappa\in\nat$. 
  Whenever $t_i$ is a set $B$, the corresponding factors in $w_1, w_2$ are \patts{p}{B}.
  By choice of $p$, it follows from a pumping argument that these
  factors can be pumped into
  \patts{\kappa}{B} in $L_1$ and $L_2$. It is then easy to check that the
  resulting words have the same pieces of size~$\kappa$. 
\end{proof}

We now move to our second step. We prove that there exists a number $\kappa$ such that two words having the same pieces of size up to $\kappa$ must both be \kimpls of a common unambiguous \ltemp (where $p$ is the number introduced in Lemma~\ref{lem:piecepump}). 
Again, we split the proof in two parts. 
We begin by proving that it is enough
to look for \ltemps for a bounded $\ell$.

\begin{lemma} \label{lem:ramseysepar}
  Let $p \in \mathbb{N}$.
  Every word is the \kimpl of some
  unambiguous \temp{N_A}, for $N_A=2^{2^{|A|}|A|(p|A|+1)}$.
\end{lemma}

\begin{proof}
  We first get rid of the unambiguity condition. Any ambiguous \ltemp
  $T$ can be reduced to an unambiguous \temp{\ell'} $T'$ with
  $\ell'<\ell$ by merging the ambiguities. It is then straightforward
  to reduce any \kimpl of $T$ into a \kimpl of $T'$. Therefore, it
  suffices to prove that every word is the \kimpl of some (possibly
  ambiguous) \temp{N_A}.

  \smallskip

  The choice of $N_A$ comes from Erdös-Szekeres' upper bound of Ramsey
  numbers. Indeed, for this value of $N_A$, from every complete graph of
  size $N_A$ with edges labeled over $2^{|A|}$ colors, there must be
  some complete monochromatic subgraph of size $p|A|+1$ (see
  \cite{bacher:hal-00763927} for a short proof that this bound
  suffices).

  \smallskip

  Observe that a word is always the \kimpl of the \temp{\ell} which is
  just the sequence of its letters. Therefore, in order to complete our
  proof, it suffices to prove that if a word is the \kimpl of some
  \temp{\ell} $T$ with $\ell > N_A$, then it is also the \kimpl of an
  \temp{\ell'} with $\ell' < \ell$.
  
  \smallskip
  Fix a word $w$, and assume that $w$ is the \kimpl of some \temp{\ell}
  $T=t_1,\dots,t_{\ell}$ with $\ell > N_A$. By definition, we get a
  decomposition $w=w_1\cdots w_{\ell}$. We construct a complete graph $\Gamma$
  with vertices $\{0,\dots, \ell\}$ and edges labeled by subsets of
  $A$. For all $i<j$, we set $\content{w_{i+1}\cdots w_{j}}$ as the
  label of the edge $(i,j)$. 
  Since $\Gamma$ has more than $\ell>N_A$
  vertices, by definition $N_A$ there exists a complete monochromatic subgraph
  with $p|A|+1$ vertices $\{i_{1},\dots,i_{p|A|+1}\}$. Let $B$ be the
  color of the edges of this monochromatic subgraph.
  Let $w'= w_{i_1+1}\cdots w_{i_{p|A|+1}}$, which is the concatenation of the
  $p|A|$ words, $w_{i_j+1}\cdots w_{i_{j+1}}$, for $j < p|A|$. By
  construction, these words have alphabet exactly $B$, hence $w'$
  is a \patt{p}{B}. It follows that $w$ is a \kimpl of the
  \temp{\ell'} $t_1,\dots,t_{i_1},B,t_{i_{p|A|+2}},\dots,t_{\ell}$
  with $\ell'= \ell - p|A| +1$. Hence $\ell'< \ell$ (except for the trivial case $p=|A|=1$).
\end{proof}

\noindent
The next lemma shown in App.~A proves that once  $\ell$ and
$p$ are fixed, given $w$ it is possible to
describe by pieces the unambiguous \ltemps that $w$ \kimpv.

\begin{lemma} \label{lem:eqimp}
  Let $\ell,p \in \mathbb{N}$. From $p$
  and $\ell$, we can compute $\kappa$ such that for every pair of words $w \peq{\kappa} w'$ and every
  unambiguous \temp{\ell} $T$, $w'$ is
  a \kimpl of $T$ whenever $w$ is a \impl{(p+1)} of $T$.
\end{lemma}

We finish the proof of the implication \ref{item:6}$\implies$\ref{item:8} by assembling
the results. Assume that $L_i$ is recognized by an NFA with $k_i$
states. Let $p=\max(k_1,k_2)$ be as introduced in Lemma~\ref{lem:piecepump},
$N_A$ as introduced in Lemma~\ref{lem:ramseysepar} for $p+1$, and
$\kappa=|A|(p+1)N_A$ be as introduced in Lemma~\ref{lem:eqimp}. Fix $\kappa' >
\kappa$ and assume that we have $w_1 \in L_1$ and $w_2 \in L_2$ such
that $w_1 \kpeq w_2$. By Lemma~\ref{lem:ramseysepar}, $w_1$ is the
\impl{(p+1)} of some unambiguous \temp{N_A} $T$. Moreover, it follows from
Lemma~\ref{lem:eqimp} that $w_2$ is a \kimpl of $T$. By
Lemma~\ref{lem:piecepump}, we finally get that there exists $w'_1 \in
L_1$ and $w'_2 \in L_2$ such that $w'_1 \peq{\kappa'} w'_2$.

The implication \ref{item:6}$\implies$\ref{item:9} of
Theorem~\ref{thm:seppiece} is easy and shown by
contraposition, see~\cite[Lemma~2]{vRZ} and
The remaining implication \ref{item:9}$\implies$\ref{item:6} can be shown using
Lemma~\ref{lem:eqimp} (see Appendix). For a direct proof, see~\cite[Lemma~3]{vRZ}, where the
key for getting a forbidden pattern out of two non-separable
languages is to extract a suitable \kimpl using Simon's factorization forests~\cite{Simon199065}.

\section{Separation by unambiguous languages}
\label{sec:FO2}

This section is devoted to proving that \UL-separability is a
decidable property. We use an argument that is analogous to
property \ref{item:8} of Theorem~\ref{thm:seppiece} in Section~\ref{sec:separation}. We prove that if
$L_1,L_2$ are languages, it is possible to compute a number $\kappa$
such that $L_1,L_2$ are \UL-separable iff they are
$\UL[\kappa]$-separable. It is then possible to test separability by
using a brute-force approach that tests all languages in $\UL[\kappa]$.

In this case, we were not able to prove equivalence between
\UL-separability and the existence of some common witness inside
the automata of both input languages. Because of this, we have a much
higher complexity. A rough analysis of the problem, which can be found
in the Appendix, gives an algorithm that runs in nondeterministic
$3$-exponential time in $\kappa$. It is likely that this can be~improved.

\noindent
We now state the main theorem of this section.
\begin{theorem}
  \label{thm:sepfod}
  Let $M_1$ and $M_2$ be monoids recognizing $L_1,L_2\subseteq A^*$. Let
  $\kappa=(2|M_1||M_2|+1)(|A|+1)^2$. Then the following
  conditions are equivalent:

  \begin{enumerate}[leftmargin=10em]
  \item\label{item:a} $L_1$ and $L_2$ are \UL-separable.
  \item\label{item:b} $L_1$ and $L_2$ are $\UL[\kappa]$-separable.
  \item\label{item:c} The language $\fdclos{L_1}{\kappa}$ separates
    $L_1$ from $L_2$.
  \end{enumerate}
\end{theorem}

As in the previous section, Condition~\ref{item:b} yields an
algorithm for testing whether two languages are separable. Indeed,
the algorithm simply computes all languages in $\UL[\kappa]$ and
checks for each of them whether it is a separator.

\begin{corollary} \label{cor:decidfod}
It is decidable whether two regular languages can be separated by a
unambiguous language.
\end{corollary}

Observe that in contrast to Theorem~\ref{thm:seppiece}, the
bound $\kappa$ of Theorem~\ref{thm:sepfod} is stated in terms of
monoids rather than in terms of automata which means an exponential
blow-up. This is necessary for our proof technique to work. 

Another remark, is that by definition of $\UL[\kappa]$, the bound
$\kappa$ is defined in terms of unambiguous products. A rephrasing of
the theorem would be: there exists a separator iff there exists one
defined by a boolean combination of unambiguous products of size
$\kappa$. It turns out that $\kappa$ also works for \fodw, \emph{i.e.},
there exists a separator iff there exists one defined by an \fodw-formula of quantifier rank $\kappa$. This can be proved by making
minor adjustements to the proof of Theorem~\ref{thm:sepfod}.

The proof of Theorem~\ref{thm:sepfod} is inspired from techniques used
in~\cite{psfo2} and relies heavily on the notion of \patts{B}{p}. The
full proof, which works by induction on the size of the alphabet, can be
found in the Appendix. We actually prove a proposition that is
basically a rephrasing of Theorem~\ref{thm:sepfod} as a pumping-like
property. In the remainder of this section, we state this proposition
and explain why it implies Theorem~\ref{thm:sepfod}.

Fix two regular languages $L_1,L_2$ over  $A$, as well as their
syntactic monoids $M_1,M_2$, and the corresponding
morphisms~$\alpha_1,\alpha_2$.

\begin{proposition}
\label{prop:fodtrans}
Let $\kappa=(2|M_1||M_2|+1)(|A|+1)^2$. For all pairs of words $w_1
\kfodeq w_2$ and all $\kappa' > \kappa$,
there exists $w'_1 \fodeq{\kappa'} w'_2$ such that
$\alpha_1(w_1)=\alpha_1(w'_1)$ and $\alpha_2(w_2)=\alpha_2(w'_2)$.
\end{proposition}

Let us briefly explain why Proposition~\ref{prop:fodtrans} implies
Theorem~\ref{thm:sepfod}. We explain why the direction
\ref{item:a}$\implies$\ref{item:c} holds, since \ref{item:c}$\implies$\ref{item:b} and
\ref{item:b}$\implies$\ref{item:a} are trivial. We prove the
contrapositive: assume that
$\fdclos{L_1}{\kappa}$ is not a separator. Hence, by definition of 
$\fdclos{L_1}{\kappa}$, there exist $w_1 \in L_1$ and $w_2 \in L_2$
such that $w_1 \kfodeq w_2$. Now, let $\kappa' > \kappa$. From
Proposition~\ref{prop:fodtrans}, it follows that there exist $w'_1 \in
L_1$ and $w'_2 \in L_2$ such that  $w'_1 \fodeq{\kappa'} w'_2$. This
means that $L_1$ and $L_2$ cannot be separated by a language of
$\UL[\kappa']$. Since this holds for all $\kappa'>\kappa$, it follows that  $L_1$ and $L_2$ are not \UL-separable.

\section{Conclusion}

We proved separation results for both piecewise testable and
unambiguous languages. Both results provide a way to decide
separability. In the \PT case, we even prove that this
can be done in {\sc Ptime}. Moreover, in both cases we give an insight
on the actual separator by providing a bound on its size should it
exist.

There remain several interesting questions in this field. First, one
could consider other subclasses of regular languages, the most
interesting one being full first-order logic. Separability by
first-order logic has already been proven to be decidable using 
semigroup theory~\cite{Henckell:1988}. However, this approach is
difficult to understand, it yields a costly algorithm, it only
provides a yes/no answer, and no insight about a possible separator. Another question is to get tight complexity bounds. For
unambiguous languages for instance, one can show a 3\textsc{Nexptime}
bound using translation to automata, but in this case, and even for
piecewise testable languages, we do not know any tight bounds on the
size of~separators.

Another observation is that right now, we have no general approach and
are bound to use \emph{ad-hoc} techniques for each subclass. An interesting
direction would be to invent a general framework that is suitable for
this problem in the same way that monoids are a suitable
framework for decidable characterizations.

\bibliographystyle{abbrv}
\bibliography{PTFO2-Arxiv}

\newpage
\section*{Appendix}
\appendix

\section*{A: Proofs of Section~\ref{sec:separation}}

\subsection*{Proof of Condition~\ref{item:9} Theorem~\ref{thm:seppiece}}

We prove Condition~\ref{item:9}. We prove
\ref{item:9}$\implies$\ref{item:6} and
\ref{item:7}$\implies$\ref{item:9}.

\medskip

\noindent
{\bf \ref{item:9}$\implies$\ref{item:6}.} We proceed by
contraposition. Assume that $L_1,L_2$ are not \PT-separable.
Recall that $\mathcal{A}_1$ and $\mathcal{A}_2$ are NFAs for $L_1,L_2$
and that $k_1,k_2$ are their sizes. Set $p=\max(k_1,k_2)+1$ and
$\ell=2^{2^{|A|}|A|(p|A|+1)}$. We prove that $L_1,L_2$ both contain
\impls{p} of some \temp{N_A} $T$ and use $T$ to construct a witness of
non \PT-separability in $(\mathcal{A}_1,\mathcal{A}_2)$.

Let $\kappa$ be as defined in Lemma~\ref{lem:eqimp} from $\ell$ and
$p$. Because $L_1,L_2$ are not \PT-separable, there exist $w_1 \in
L_1$ and $w_2 \in L_2$ such that $w_1 \peq{\kappa} w_2$. By choice of
$\ell,p$ and Lemma~\ref{lem:ramseysepar}, $w_1$ must be the \impl{p} of
some unambiguous \temp{\ell} $T$. Applying Lemma~\ref{lem:eqimp}, we
obtain that $w_1,w_2$ are both \impls{(\max(k_1,k_2))} of $T$.

Let $\vec B=(B_1,\ldots,B_n)$ be the subsequence of elements $T$ that
are sets. Let $\vec u=(u_0,\ldots,u_n)$, where $u_i$ is the word
obtained by concatenating the letters that are between $B_i$ and $B_{i+1}$
in $T$.  By definition $(\vec u, \vec B)$ is a factorization pair.

Because $w_1$ is a \impl{(\max(k_1,k_2))}, the path used to read $w_1$
in $\mathcal{A}_1$  must traverse loops labeled by each of the $B_i$,
and clearly this is a $(\vec u,\vec B)$-path. Using the same argument
we get that the path of $w_2$ in $\mathcal{A}_2$ is also a $(\vec
u,\vec B)$-path. Therefore $(\vec u, \vec B)$ is a witness of non
\PT-separability.

\medskip

\noindent
{\bf \ref{item:7}$\implies$\ref{item:9}.} Again, we proceed
by contraposition. Set $\kappa$ as in Theorem~\ref{thm:seppiece} and
assume that there exists a factorization pair  $(\vec u, \vec B)$ that
is a witness of non \PT-separabilty. We prove that
$\ptclos{L_1}{\kappa}$ is not a separator.

Set $\vec B=(B_1,\ldots,B_n)$ and $\vec u=(u_0,\ldots,u_n)$. By
definition, this means that there exists $w_1 \in L_1$ and $w_2 \in
L_2$ of the form

\[
u_0v_iu_1v_1 \cdots v_nu_n,
\]

\noindent
where the words such that $\content{v_i}=B_i$ and $v_i$ contains a
$\patt{B_i}{\kappa}$. It is straightforward to see that $w_1 \kpeq
w_2$. Therefore, $w_2 \in L_2 \cap \ptclos{L_1}{\kappa}$ and
$\ptclos{L_1}{\kappa}$ is not a separator.

\subsection*{Proof of Lemma~\ref{lem:eqimp}}

\setcounter{lemma}{5}
\begin{lemma} \label{lem:eqimp2}
  Let $\ell,p \in \mathbb{N}$. From $p$
  and $\ell$, we can compute $\kappa$ such that for every pair of words $w \peq{\kappa} w'$ and every
  unambiguous \temp{\ell} $T$, $w'$ is
  a \kimpl of $T$ whenever $w$ is a \impl{(p+1)} of $T$.
%
 % Let $\ell,p$ be integers, then there exists $\kappa$ computable from $p$
 % and $\ell$ such that for all two words $w \peq{\kappa} w'$ and all
 % unambiguous \temps{\ell} $T$, if $w$ is a \impl{(p+1)} of $T$, then $w'$ is
 % a \kimpl of $T$.
\end{lemma}

\begin{proof}
  We prove that the lemma holds for $\kappa=|A|p\ell$. 
   Let $w \peq{\kappa}
  w'$ and let $T=\{t_1,t_2,\dots,t_\ell\}$ be an unambiguous \temp{\ell} such that
  $w$ is a \impl{p+1} of $T$.
  %Set $w \peq{\kappa}
  %w'$ and $T=\{t_1,t_2,\dots,t_\ell\}$ an unambiguous \temps{\ell} such that
  %$w$ is a \impl{p+1} of $T$. 
  We begin by giving a decomposition of $w'$
  and prove that it indeed witnesses the fact that $w'$ is a \kimpl of $T$.
  We define $w'_1\cdots w'_\ell=w'$ inductively as follows: assume that the
  factors are defined up to $w'_i$ and let $u$ be such that $w'=w_1\cdots
  w'_i \cdot u$. If $t_{i+1}$ is a letter then $w_{i+1}$ is just the
  first letter of $u$, otherwise $t_{i+1}=B \subseteq A$, and in that case
  $w_{i+1}$ is the largest prefix of $u$ which contains only letters of
  $B$. We will prove that $w'_1\cdots w'_\ell$ witnesses that
  $w'$ is a \kimpl of $T$. The proof relies on a subresult that we state
  and prove~below.

  To every unambiguous \temp{\ell} $T=\{t_1,t_2,\dots,t_\ell\}$ and %integer
  $p \in \mathbb{N}$, we associate a piece $v_{T,p}=v_1\cdots v_\ell$,
  such that for all $i$: 
  \[
  v_i = \left\{\begin{array}{ll} a & \text{if $t_i=a \in A$}\\(b_1 \cdots
      b_n)^p & \text{if $t_i=\{b_1,\dots,b_n\}\subseteq A$}\end{array}\right.
  \]
  By definition, if $w$ is a \kimpl of $T$ then $v_{T,p}$ is a piece of
  $w$. Consider $v_{T,1}=v_1\cdots v_\ell$. A piece $v$ is
  \emph{incompatible} with $T$ when $v$ is of the following form:
  $
  v = v_1\cdots v_i \cdot u \cdot v_{i+1} \cdots v_\ell
  $
  such that if $t_i$ (resp. $t_{i+1}$) is a set the first letter
  (resp. last letter) of $u$ is not in $t_{i}$ (resp. $t_{i+1}$).

  \begin{clm} \label{clm:unambig}
    %Assume 
    If $w$ is a \impl{1} of some unambiguous \temp{\ell} $T$, then there
    is no piece of $w$ that is incompatible with $T$.
  \end{clm}

  \begin{proof}
    This is a consequence of the fact that $T$ is unambiguous.
  \end{proof}

  We now prove that for all $i$, $w'_i=a$ if $t_i$ is the letter $a$
  or $w'_i$ is a word over some $B\subseteq A$ containing a \patt{p}{B}
  if $t_i=B$. We proceed by induction, assume that this is true up to
  $w'_i$ and consider $w'_{i+1}$. Set $v_{T,1} = v_1 \cdots v_\ell$ and
  $v_{T,p+1}=v_1 \cdots v_\ell$. By induction hypothesis and by choice of
  $\kappa$, we know that $v_{i+1} \cdots v_\ell$ and $v_{i+1} \cdots v_\ell$
  are pieces of $w_{i+1} \cdots w_\ell$. We distinguish two cases depending
  on the nature of $t_{i+1}$.

  \medskip
  \noindent
  {\it Case 1: $t_{i+1}$ is some letter $a$.} We have to prove that
  $w'_{i+1}=a$. Assume that $w'_{i+1} = b \neq a$. Then $v_1 \cdots v_i
  \cdot b \cdot v_{i+1} \cdots v_\ell$ must be a piece of $w'$ and
  therefore a piece of $w$ ($w \peq{\kappa} w'$). We prove that this
  piece is incompatible with $T$, which contradicts
  %which is a contradiction to
  Claim~\ref{clm:unambig}. Observe that by definition of $w_i$, if $t_i$
  is a set then $b \not\in t_i$ (otherwise it would have been included
  in $t_i$). Therefore $v_1 \cdots v_i \cdot b \cdot v_{i+1} \cdots v_\ell$
  is incompatible with $T$ and we are done for this case.

  \medskip
  \noindent
  {\it Case 2: $t_{i+1}$ is a set $B=\{b_1,\dots,b_n\}$.} By
  construction, $w_{i+1}$ contains only letters in $B$. Therefore, we
  have to prove  that it contains a \patt pB. Assume that it does
  not. By contruction, the first letter of $w_{i+2}$ is some letter
  $c \not\in B$. If $w_{i+1} = \varepsilon$, then $v_1 \cdots v_i
  \cdot c \cdot v_{i+1} \cdots v_\ell$ must be a piece of $w'$. Using
  a similar argument to the one of the previous case, we can prove
  that this piece is incompatible with $T$, which contradicts 
  Claim~\ref{clm:unambig}.

  Otherwise, let $b$ be the first letter of $w_{i+1}$. Recall that by
  definition, $r_{i+1}$ contains a \patt{p+1}{B} and that $r_{i+1}
  \cdots r_\ell$ is a piece of $w_{i+1} \cdots w_\ell$. Therefore, since
  $w_{i+1}$ does not contain a \patt{p}{B}, the last suffix of $r_{i+1}$
  containing a \patt{1}{B} must fall in $w_{i+2}$ and consequently,
  $v_{i+1} \cdots v_\ell$ must be a piece of $w_{i+2} \cdots w_\ell$. It
  follows that $v_1 \cdots v_i \cdot b \cdot c \cdot v_{i+1} \cdots
  v_\ell$ is a piece of $w'$ and of $w$ ($w \peq{\kappa} w'$). By definition,
  $c \not\in t_{i+1}$. Moreover, by construction, if $t_i$ is a set,
  then $b \not\in t_i$. Therefore, $v_1 \cdots v_i \cdot b \cdot c \cdot
  v_{i+1} \cdots v_\ell$ is incompatible with $T$ which contradicts
  Claim~\ref{clm:unambig} since it is also a piece of $w$.
\end{proof}

\section*{B: Proofs of Section~\ref{sec:FO2}}
\label{sec:a1}

\subsection*{Complexity Analysis}

We briefly explain how \UL-separability can be checked in
3\textsc{Nexptime}. This is a very rough analysis and it is likely
that this can be improved. We use the connection between \UL and
\fodw to prove that any language in $\UL[\kappa]$ is recognized by
a deterministic automaton that is $3$-exponential in $\kappa$. In
order to check separability it is then enough to non-deterministically
guess an automaton of size $3$-exponential in $\kappa$, check if the
automata recognizes an unambiguous language and check if it is a
separator. This can be done in non-deterministic $3$-exponential
time.

It is straightforward to see that any unambiguous product can be
described by an \fodw formula whose nesting depth of quantifiers is
polynomial in the size $\kappa$ of the product. One can then construct
an equivalent unary temporal logic ({\bf UTL}) formula  that is
exponential in size~\cite{evwutl}. From this {\bf UTL} formula, it is
then possible to construct an equivalent deterministic automaton that
is double exponential in size~\cite{dgtemporal}. This automaton is
$3$-exponential in $\kappa$.

\subsection*{Proof of Theorem~\ref{thm:sepfod}}

We prove Proposition~\ref{prop:fodtrans}. As we explained, our proof
relies heavily on the notion of \patts{B}{p}. However, in this case we
will only need to use \patts{B}{p} where $B$ is the whole alphabet
$A$. For this reason, we simply write \pata{p} for \patt{A}{p}.
We begin by giving a brief outline
of the proof. We fix a large enough $l$, intuitively $l$ needs to be
large enough so that we can apply our pumping arguments to words of
length $l$. Then, we show that $\kappa$ is large enough in order to
ensure that $w_1,w_2$ are both a \pata{2l} or neither of them are. In
both cases, we are then able to decompose $w_1,w_2$ into sequences of
factors that are pairwise equivalent and use a smaller alphabet. We
then apply induction on these factors. In the second case, it then
suffices to reconcatenate the factors to obtain the desired result. In
the first case, a pumping argument depending on $l$  will also be
needed. Before providing the lemmas involved in the construction, we
first give some additional definitions concerning \kpatas that will
come in conveniently.
  
\medskip
\noindent
{\bf Decompositions for \kpatas.} Recall that by definition, a word
$w$ is a \kpata iff $w \in (A^*a_1 \cdots A^*a_nA^*)^k$. This means
that $w$ can be decomposed into a sequence of factors witnessing its
membership in $(A^*a_1 \cdots A^*a_nA^*)^k$:
\[
w=\prod_{i=1}^{kn}(w_i \cdot a_{i\ mod\ n}) \cdot w_{kn+1}.
\]
To each word $w$ we associate a unique such decomposition for the
largest integer $k$ such that $w$ is a \kpata. Let $w \in A^*$, we say
that $w$ \emph{admits a \kdecop} iff $w$ is a \kpata but not a
\pata{k+1}. It is straightforward to see that if $w$ admits a \kdecop
then there exists an integer $l$ such that $kn < l < (k+1)n$,~and
\begin{equation} \label{eq:decompl}
w=\prod_{i=1}^{l}(w_i \cdot a_{i\ mod\ m}) \cdot w_{l+1},
\end{equation}
such that for all $i$, $a_{i\ mod\ m}\not\in w_i$. The integer $l$ is
called the \emph{length} of the decomposition. We give an example of a
word that admits a \decop{1} (of length $5$) in Figure~\eqref{fig:kdecomp}.
%\vspace*{-5mm}
\begin{figure}[h]
  \begin{center}
    \begin{tikzpicture}
      \node (w0) at (0.0,0.0) {bc{\bf a}c{\bf b}b{\bf c}cc{\bf a}cc{\bf b}aa};
      \node[rotate=90] (p1) at ($(w0.south)-(1.08,0.0)$) {$\{$};
      \node at ($(p1)-(0.0,0.2)$) {$w_1$};
      \node[rotate=90,yscale=0.5] (p2) at ($(w0.south)-(0.66,0.0)$) {$\{$};
      \node at ($(p2)-(0.0,0.2)$) {$w_2$};
      \node[rotate=90,yscale=0.5] (p3) at ($(w0.south)-(0.30,0.0)$) {$\{$};
      \node at ($(p3)-(0.0,0.2)$) {$w_3$};
      \node[rotate=90] (p4) at ($(w0.south)+(0.1,0.0)$) {$\{$};
      \node at ($(p4)-(0.0,0.2)$) {$w_4$};
      \node[rotate=90] (p5) at ($(w0.south)+(0.59,0.0)$) {$\{$};
      \node at ($(p5)-(0.0,0.2)$) {$w_5$};
      \node[rotate=90] (p6) at ($(w0.south)+(1.10,0.0)$) {$\{$};
      \node at ($(p6)-(0.0,0.2)$) {$w_6$};
    \end{tikzpicture}
  \end{center}
  % \vspace*{-5mm}
  \caption{$w=bcacbbcccaccbaa$ over $A=\{a,b,c\}$ admits a \decop{1}}
  \label{fig:kdecomp}
\end{figure}
%\vspace*{-5mm}

We now state two lemmas that are needed for our construction. Both
lemmas link \decops{k} to unambiguous languages. The first one states
that \kdecops can be detected using an unambiguous product with large
enough size. Moreover, using products of the same size, one can also
describe the products that are satisfied by the factors in the \kdecop.

\begin{lemma} \label{lem:foddecompos}
Let $k, \tilde{\kappa} \in \mathbb{N}$ such that $\kappa >
\tilde{\kappa} + k(|A|+1)$. Then for every pair of words $u \kfodeq v$
and all $h \leq k$, $u$ admits an \decop{h} iff $v$ admits an
\decop{h}. Moreover, the associated decompositions, as described
in~\eqref{eq:decompl}, are of the same length $l$:
\[
\begin{array}{lcl}
u & = & \prod_{i=1}^{l}(u_i \cdot a_{i\ mod\ m}) \cdot u_{l+1} \\
v & = & \prod_{i=1}^{l}(v_i \cdot a_{i\ mod\ m}) \cdot v_{l+1}
\end{array}
\]
and for all $i$, $u_i \fodeq{\tilde{\kappa}} v_i$.
\end{lemma}

\begin{proof}
For the sake of readability, for all $a \in A$ we will write $A_{a}$
for the alphabet $A \setminus \{a\}$ and for all number $i$, we will
write $b_i$ for $a_{i\ mod\ m}$. Assume that $u$ admits an \decop{h} and
consider the associated decomposition:

\[
u = \prod_{i=1}^{l}(u_i \cdot b_i) \cdot u_{l+1}.
\]

Consider the following unambiguous product:

\[
P = \prod_{i=1}^{l}(A_{b_{i}} \cdot b_{i})
\cdot A_{b_{l+1}}.
\]

By construction, $P$ is of size $l \leq k(|A|+1) < \kappa$ and $u \in
P$. Therefore, $v \in P$ and $v$ admits an \decop{h} together with the
associated decomposition: 

\[
v = \prod_{i=1}^{l}(v_i \cdot b_{i}) \cdot v_{l+1}.
\]

It remains to prove that for all $i$, $u_i \fodeq{\tilde{\kappa}}
v_i$. Assume that $u_i$ belongs to some unambiguous product $P'$ of
size $\tilde{\kappa}$. This means that $u$ is in the unambiguous
product:

\[
P'' = \prod_{j=1}^{i-1}(A_{b_{j}} \cdot b_{j})
\cdot (P' \cdot b_{i}) \cdot \prod_{j=i+1}^{l}(A_{b_{j}}\cdot b_j) \cdot A_{b_{l+1}}.
\]

Because $P''$ is of size at most $\tilde{\kappa} + k(|A|+1)$, we have
$v \in P''$. By assumption on the decomposition of $v$, this means
that $v_i \in P'$ and we are done.
%Straightforward by using \efgame games for \fodw.
\end{proof}

%Our second lemma states that if a word has both a prefix and a suffix
%that are large enough \kpatas, then \fodw cannot express any property
%about the middle part of the word.

Our second lemma states that if a word has both a prefix and a suffix
that are large enough \kpatas, then the middle part has no influence on
membership in an unambiguous product.

\begin{lemma} \label{lem:pattfo2}
Let $\kappa' \in \mathbb{N}$, and let $u,v$ be two words that are both
\patas{\kappa'}. Then for all words $w_1,w_2$, the following
equivalence holds:
\[
u \cdot w_1 \cdot v \fodeq{\kappa'} u \cdot w_2 \cdot v.
\]
\end{lemma}

\begin{proof}
We begin by observing a simple property of unambiguous products.

\begin{myremark}  \label{rem:unam}
Let $B_0^*a_1B_1^*\cdots B_{\kappa'-1}^*a_{\kappa'}B_{\kappa'}^*$ be an
unambiguous product. There can be at most one set $B_i$ such that
$B_i = A$.
\end{myremark}

Let $P=B_0^*a_1B_1^*\cdots B_{\kappa'-1}^*a_{\kappa'}B_{\kappa'}^*$ an
unambiguous product of size $\kappa'$ and assume that $uw_1v \in
P$. We prove that $uw_2v \in P$. Since $uw_1v \in P$, there exists
some decomposition

\[
uw_1v=x_0a_1x_1 \cdots x_{\kappa'-1}a_{\kappa'}x_{\kappa'}
\]

\noindent
that is a witness. If no set $B_i$ is the whole alphabet $A$, then
$uw_1v$ is at most a \pata{\kappa'}, which is impossible since it is
by definition a \pata{2\kappa'}. Therefore, by Remark~\ref{rem:unam}
there is exactly one set $B_i$ such that $B_i=A$. It follows
that the words $x_0a_1x_1 \cdots x_{i-1}$ and $a_{i+1} \cdots
x_{\kappa'-1}a_{\kappa'}x_{\kappa'}$ are at most
\patas{\kappa'-1}. Therefore, they are respectively a prefix of $u$
and a suffix of $v$ (which are \patas{\kappa'}). It follows that there
exists some word $y_i$ such that

\[
uw_2v=x_0a_1x_1 \cdots x_{i-1}y_ia_{i+1} \cdots
x_{\kappa'-1}a_{\kappa'}x_{\kappa'}.
\]

Such a decomposition for $uw_2v$ is witness for membership in $P$. We
conclude that $uw_2v \in P$.

%Again by a standard \efgame argument for \fodw.
\end{proof}

We can now finish our construction by using
Lemma~\ref{lem:foddecompos} and Lemma~\ref{lem:pattfo2} and prove
Proposition~\ref{prop:fodtrans}. The proof goes by induction on the
size of the alphabet.

We begin by fixing the size of the patterns we are going to look for
in $w_1,w_2$.  Set $k = |M_1||M_2|$. A pigeonhole principle argument
proves that for $s_0,\dots,s_k \in M_1$ and $r_0,\dots,r_k \in M_2$,
there exist $i<j$ such that both $s_0\cdots s_{i-1}=s_0\cdots s_{j}$
and $r_0\cdots r_{i-1}=r_0\cdots r_{j}$. We will look for \patas{k}.

Observe that $\kappa = (2k+1)(m+1)^2$ (recall that $m = |A|$). We
prove Proposition~\ref{prop:fodtrans} by induction on $m$.
Let $\kappa' > \kappa$, $w_1 \kfodeq w_2$ and set
$\tilde{\kappa}=(2k+1)m^2$. By Lemma~\ref{lem:foddecompos}, either
both $w_1,w_2$ admit \decops{h} for some $h < 2k$ or both $w_1,w_2$ do not admit \decop{h} for $h \leq 2l$. We treat these cases~separately.
 
\medskip\noindent
{\it \bf Case 1: both $w_1$ and $w_2$ admit \decops{h} for $h < 2k$.}
%In that case, set $\tilde{\kappa}=(2k+1)m^2$. 
Observe that $\tilde{\kappa}=(2k+1)m^2$ implies that $\kappa >
\tilde{\kappa} + 2k(m+1)$. Therefore, we can apply the second part of
Lemma~\ref{lem:foddecompos} to $w_1$ and $w_2$. This yields
the following decompositions:
\[
\begin{array}{lcl}
w_1 & = & \prod_{i=1}^{l}(u_i \cdot a_{i\ mod\ m}) \cdot u_{l+1} \\
w_2 & = & \prod_{i=1}^{l}(v_i \cdot a_{i\ mod\ m}) \cdot v_{l+1}
\end{array}
\]
such that for all $i$, $u_i \fodeq{\tilde{\kappa}} v_i$. Also observe
that by definition of these decompositions $u_i,v_i$ use a strict
subalphabet of $A$. Therefore, the induction hypothesis can be used
and for all $i$ we get words $u'_i,v'_i$ such that $u'_i
\fodeq{\kappa'} v'_i$, $\alpha_1(u_i)=\alpha_1(u'_i)$ and
$\alpha_2(v_i)=\alpha_2(v'_i)$. Now, consider the words:
\[
\begin{array}{lcl}
w'_1 & = & \prod_{i=1}^{l}(u'_i \cdot a_{i\ mod\ m}) \cdot u'_{l+1} \\
w'_2 & = & \prod_{i=1}^{l}(v'_i \cdot a_{i\ mod\ m}) \cdot v'_{l+1}
\end{array}
\]
By construction, we have $w'_1 \fodeq{\kappa'} w'_2$,
$\alpha_1(w_1)=\alpha_1(w'_1)$ and $\alpha_2(w_2)=\alpha_2(w'_2)$ and
we are done.

\medskip
\noindent
{\it \bf Case 2: both $w_1$ and $w_2$ do not admit \decops{h} for $h \leq 2k$.}
This means that $w_1,w_2$ are both \patas{2k}. A simple argument as in
the proof of Lemma~\ref{lem:foddecompos} shows that
$w_1,w_2$ can be decomposed into three factors:
%for \fodw shows that $w_1,w_2$ can be decomposed into three
%factors:

\newcommand\frl{\ensuremath{\mathfrak{l}}\xspace}
\newcommand\frr{\ensuremath{\mathfrak{r}}\xspace}
\newcommand\frc{\ensuremath{\mathfrak{c}}\xspace}

\[
\begin{array}{lcl}
w_1 & = & u_{\frl} \cdot u_{\frc} \cdot u_{\frr} \\
w_2 & = & v_{\frl} \cdot v_{\frc} \cdot v_{\frr}
\end{array}
\]

such that $u_{\frl} \fodeq{\kappa - k(m+1)} v_{\frl}$, $u_{\frr} \fodeq{\kappa -
  k(m+1)} v_{\frr}$, $u_{\frc} \fodeq{\kappa - k(m+1)} v_{\frc}$ and
$u_{\frl},v_{\frl},u_{\frr},v_{\frr}$ admit \kdecops. We prove that it is possible to
construct $u'_{\frl},v'_{\frl},u'_{\frr},v'_{\frr}$  
%$u'_{\frl},v'_{\frr},u'_{\frl},v'_{\frr}$ 
such that:

\begin{enumerate}
\item $u'_{\frl} \fodeq{\kappa'} v'_{\frl}$ and $u'_{\frr}
  \fodeq{\kappa'} v'_{\frr}$.
\item $\alpha_1(u_{\frl})=\alpha_1(u'_{\frl})$,
  $\alpha_1(u_{\frr})=\alpha_1(u'_{\frr})$,
  $\alpha_2(v_{\frl})=\alpha_2(v'_{\frl})$ and
  $\alpha_2(v_{\frr})=\alpha_2(v'_{\frr})$.
\item $u'_{\frl},v'_{\frl},u'_{\frr},v'_{\frr}$ are all \patas{\kappa'}.
\end{enumerate}

We only give the proof for $u'_{\frl},v'_{\frl}$, as the proof for
$u'_{\frr},v'_{\frr}$ is identical. Observe that $\kappa - k(m+1) >
\tilde{\kappa} + k(m+1)$. Therefore, we can apply Lemma~\ref{lem:foddecompos}
to $u_{\frl}$ and $v_{\frl}$. This yields the following decompositions: 
\[
\begin{array}{lcl}
u_{\frl} & = & \prod_{i=1}^{l}(u_i \cdot a_{i\ mod\ m}) \cdot u_{l+1} \\
v_{\frl} & = & \prod_{i=1}^{l}(v_i \cdot a_{i\ mod\ m}) \cdot v_{l+1}
\end{array}
\]
such that for all $i$, $u_i \fodeq{\tilde{\kappa}} v_i$. Moreover the
length $l$ of the decompositions is $l > km$. Also observe
that by definition of these decompositions, $u_i,v_i$ use a strict
subalphabet of $A$. Therefore, the induction hypothesis can be used
and for all $i$ we get words $u'_i,v'_i$ such that $u'_i
\fodeq{\kappa'} v'_i$, $\alpha_1(u_i)=\alpha_1(u'_i)$ and
$\alpha_2(v_i)=\alpha_2(v'_i)$. We now use the fact that our choice of
$k$ allows us to pump the sequences of factors into large
sequences: there exist numbers $0 \leq j_1<j_2 \leq
k$ such~that
\[
\begin{array}{lcl }
\prod_{i=1}^{(j_1-1)n}\alpha_1(u'_i \cdot a_{i\ mod\ m}) & = & \prod_{i=1}^{j_2n}\alpha_1(u'_i \cdot a_{i\ mod\ m})\\
\prod_{i=1}^{(j_1-1)n}\alpha_2(v'_i \cdot a_{i\ mod\ m}) & = & \prod_{i=1}^{j_2n}\alpha_1(u'_i \cdot a_{i\ mod\ m})
\end{array}
\]
Now set,
\[
\begin{array}{lcl}
e_1 & = & \prod_{i=(j_1-1)n+1}^{j_2n}u'_i \cdot a_{i\ mod\ m}\\
e_2 & = & \prod_{i=(j_1-1)n+1}^{j_2n}v'_i \cdot a_{i\ mod\ m}
\end{array}
\]

\[
\begin{array}{lcl}
u'_{\frl}   & = & \prod_{i=1}^{j_1n}(u'_i \cdot a_{i\ mod\ m}) \cdot (e_1)^{\kappa'}
\cdot \prod_{i=j_2n+1}^{l}(u'_i \cdot a_{i\ mod\ m}) \cdot u'_{l+1}
\\
v'_{\frl}   & = & \prod_{i=1}^{j_1n}(v'_i \cdot a_{i\ mod\ m}) \cdot (e_2)^{\kappa'}
\cdot \prod_{i=j_2n+1}^{l}(v'_i \cdot a_{i\ mod\ m}) \cdot v'_{l+1} 
\end{array}
\]

Observe that by construction $u'_{\frl} \fodeq{\kappa'} v'_{\frl}$
(they are products 
%sums 
of pairwise equivalent factors). Moreover, by construction of
$e_1,e_2$, we have
$\alpha_1(u_{\frl})=\alpha_1(u'_{\frl})$ and
$\alpha_2(v_{\frl})=\alpha_1(v'_{\frl})$. Finally, since $e_1,e_2$
are \patas{1}, $u'_{\frl},v'_{\frl}$ are \patas{\kappa'}. Using the
same construction for $u_{\frr},v_{\frr}$, we obtain
$u'_{\frr},v'_{\frr}$. We can now finish our construction by giving
$w'_1,w'_2$:

\[
\begin{array}{lcl}
w'_1 & = & u'_{\frl} \cdot u_{\frc} \cdot u'_{\frr} \\
w'_2 & = & v'_{\frl} \cdot v_{\frc} \cdot v'_{\frr}
\end{array}
\]

By construction it is clear that $\alpha_1(w_1)=\alpha_1(w'_1)$ and
$\alpha_2(w_2)=\alpha_2(w'_2)$. Moreover, because
$u'_{\frl},u'_{\frr}$ are \patas{\kappa'} it follows from
Lemma~\ref{lem:pattfo2} that
\[
u'_{\frl} \cdot u_{\frc} \cdot u'_{\frr} \fodeq{\kappa'} u'_{\frl} \cdot v_{\frc} \cdot u'_{\frr}
\]

\noindent
Finally, since $u'_{\frl} \fodeq{\kappa'} v'_{\frl}$ and $u'_{\frr}
\fodeq{\kappa'} v'_{\frr}$, we have
\[
u'_{\frl} \cdot v_c \cdot u'_{\frr} \fodeq{\kappa'} v'_{\frl} \cdot v_c \cdot v'_{\frr}
\]
Combining the two equivalences we obtain $w'_1 \fodeq{\kappa'} w'_2$,
which concludes the proof.

%%% Local Variables:
%%% TeX-view-style: "/Applications/Skim.app/Contents/SharedSupport/displayline %n %o %b"
%%% TeX-PDF-mode: t
%%% mode: latex
%%% TeX-master: "SeparationApril'13"
%%% End:

\end{document}